\documentclass[12pt,twoside]{article}

\usepackage[a4paper]{geometry}
\usepackage{amsmath}
\usepackage{mathtools}
\usepackage{lipsum}
\usepackage{float}
\usepackage{microtype}
\usepackage[font=small,labelfont=bf]{caption}
\usepackage{fancyref}
\usepackage{color,soul} 
\usepackage{enumerate}
\usepackage{tabularx}
\usepackage{amsthm}
\usepackage{tikz}
\usetikzlibrary{shapes,arrows}
\usepackage{smartdiagram}
\usepackage{mathrsfs}  
\usepackage{amssymb}
\usepackage{graphicx}
\usepackage{verbatim}
\usepackage{algorithm}
\usepackage{algorithmic}

\makeatletter
\renewcommand\title[1]{\gdef\@title{\reset@font\Large\bfseries #1}}
\renewcommand\section{\@startsection {section}{1}{\z@}%
                                   {-3.5ex \@plus -1ex \@minus -.2ex}%
                                   {2.3ex \@plus.2ex}%
                                   {\normalfont\large\bfseries}}
\renewcommand\subsection{\@startsection{subsection}{2}{\z@}%
                                     {-3ex\@plus -1ex \@minus -.2ex}%
                                     {1.5ex \@plus .2ex}%
                                     {\normalfont\normalsize\bfseries}}
\renewcommand\subsubsection{\@startsection{subsubsection}{3}{\z@}%
                                     {-2.5ex\@plus -1ex \@minus -.2ex}%
                                     {1.5ex \@plus .2ex}%
                                     {\normalfont\normalsize\bfseries}}

\def\@runningauthor{}\newcommand{\runningauthor}[1]{\def\runningauthor{#1}}
\def\@runningtitle{}\newcommand{\runningtitle}[1]{\def\runningtitle{#1}}

\renewcommand{\ps@plain}{%
\renewcommand{\@evenhead}{\footnotesize\scshape \hfill\runningauthor\hfill}
\renewcommand{\@oddhead}{\footnotesize\scshape \hfill\runningtitle\hfill}}
\g@addto@macro\bfseries{\boldmath}

\makeatother

\usepackage{amsthm,amsmath,amssymb}
\usepackage{mathrsfs}  

\usepackage{graphicx}

\usepackage[colorlinks=true,citecolor=black,linkcolor=black,urlcolor=blue]{hyperref}

\font\msbm=msbm10 at 10pt

\newcommand{\NN}{\mbox{\msbm N}}

\newcommand{\FF}{\mbox{\msbm F}}

\def \N {{\NN}}
\def \F {{\FF}}
\def \x {{\bf x}}
\def \y {{\bf y}}

\theoremstyle{plain}
\newtheorem{theorem}{Theorem}
\newtheorem{lemma}[theorem]{Lemma}

\newtheorem{proposition}[theorem]{Proposition}
\newtheorem{fact}[theorem]{Fact}

\theoremstyle{definition}
\newtheorem{definition}[theorem]{Definition}

\theoremstyle{remark}
\newtheorem{remark}[theorem]{Remark}

\title{On Universally Good Flower Codes}

\runningtitle{On Universally Good Flower Codes}

\author{Krishna Gopal Benerjee 
\\
\small Laboratory of Natural Information Processing,\\[-0.8ex]
\small Dhirubhai Ambani Institute of Information and Communication Technology,  \\[-0.8ex] 
\small Gandhinagar, Gujarat, India\\ 
\small\tt krishna\_gopal@daiict.ac.in\\
\and
 Manish K Gupta 
 \\
 \small Laboratory of Natural Information Processing,\\[-0.8ex]
 \small Dhirubhai Ambani Institute of Information and Communication Technology, \\[-0.8ex] 
 \small Gandhinagar, Gujarat, India\\
 \small\tt mankg@computer.org
}

\date{}

\begin{document}

\maketitle

\thispagestyle{empty}

\begin{abstract}
For a Distributed Storage System (DSS), the \textit{Fractional Repetition} (FR) code is a class in which replicas of encoded data packets are stored on distributed chunk servers, where the encoding is done using the Maximum Distance Separable (MDS) code. 
The FR codes allow for exact uncoded repair with minimum repair bandwidth.
In this paper, FR codes are constructed using finite binary sequences. 
The condition for universally good FR codes is calculated on such sequences. 
For some sequences, the universally good FR codes are explored. 
\end{abstract}

\section{Introduction} 
The distributed storage is a well studied area, which deals with storing the data on distributed nodes in such a manner so that it allows data accessibility at anytime and anywhere. 
Many companies such as Microsoft and Google etc. provide such storage services by using distributed data centers.
In the storage systems like Google file system \cite{Ghemawat03},
multiple copies of data fragments are stored.
Thus, the system becomes reliable and the file can be retrieved from the system. 
At the same level of redundancy, coding techniques can improve the reliability for the storage system \cite{Weatherspoon:2002:ECV:646334.687814}. 

In Distributed Storage Systems (DSSs), a data file is encoded to a certain number of packets of the same size and those packets are stored on $n$ distinct nodes.
To retrieve the complete data file information from the DSS, a data collector has to download packets from any $k$ (called reconstruction degree) nodes. 
For a reliable system, a failed node has to be repaired in the DSS. 
For the repair, a failed node is replaced by a new active node with the same information.
For the new node, the packets are downloaded from each node (called helper node) of a set of $d$ (called repair degree) active nodes. 
The set can be called surviving set.
The repair is called exact (or functional), if the packets of the new node are the exact copy (or the function) of the packets of the failed node. 
For the repair, if $\beta$ packets are downloaded from a helper node of the surviving set then the total repair bandwidth will be $d\beta$.

For an $(n, k, d)$ DSS, the regenerating codes are given by the parameters  $\{[n, k, d], [\alpha,\beta, B] \}$, where the data file is broken into B information packets and a node contains $\alpha$ encoded packets \cite{5709963}. 
Optimizing both the parameters $\alpha$ and $\beta$ in different order, one can get two kinds of regenerating codes viz. Minimum Storage Regenerating (MSR) and Minimum Bandwidth Regenerating (MBR) codes. 
The MSR codes are useful for archival purpose and the MBR codes are useful for Internet applications. 
Some of the MBR codes studied by the researchers fail to optimize other parameters of the system such as disk I/O, computation and scalability etc. 
Towards this goal, the disk I/O is optimized for a class of MBR codes called Dress codes \cite{rr10,6062413,RSKR09,RSKR10}.
For a failed node, the exact copies of packets are downloaded from the helper nodes and placed at the new node.
The exact repair reduces the computation cost and is known as repair-by-transfer or table based repair or encoded repair. 
Dress codes are encoded by two-layer code, the inner code called fractional repetition (FR) code and the outer MDS code \cite{rr10,6033980}. 
The FR codes are constructed using graphs \cite{7118709}, combinatorial designs \cite{7118709,6912604,6763122,7004501,7066224,7422071,6483351}, difference families  \cite{e19100563,e18120441} and other combinatorial configurations  \cite{6810361,6033980,DBLP:journals/corr/abs-1208-2787}. 
FR codes have been studied in different directions such as Weak FR codes \cite{DBLP:journals/corr/abs-1302-3681,DBLP:journals/corr/abs-1303-6801}, Irregular FR codes \cite{6804948}, Variable FR codes \cite{6811237} General Fractional Repetition codes \cite{6763122}, Heterogeneous Fractional Repetition codes \cite{iet:/content/journals/10.1049/iet-com.2014.1225}, Locally Repairable  Fractional Repetition codes \cite{7458387,7558231}, Adaptive Fractional Repetition codes \cite{7312417}, Scalable Fractional Repetition codes \cite{6120326} and others \cite{7366761}. 
For a given FR code, algorithms for repair degree and reconstruction degree are given in \cite{DBLP:journals/corr/abs-1305-4580,DBLP:journals/corr/PrajapatiG16}. 
A ring construction of FR codes was described in \cite{DBLP:journals/corr/abs-1302-3681} (for $\rho=2$) and \cite{7458383} (for $\rho\geq2$), where nodes are placed on a circle and packets are dropped on nodes in a specific manner. 

\textit{Contribution}: 
FR codes with non-uniform parameters are more close to the real world. 
Motivated by \cite{DBLP:journals/corr/abs-1302-3681,7458383}, in this paper, FR codes with non-uniform parameters, are constructed using finite binary sequences.   
For such sequences, we have calculated parameters and existing conditions of the corresponding FR codes. 
In general, for a universally good FR code, a bound on such sequences is calculated. 
We have obtained universally good FR code for periodic binary sequences.

\textit{Organization}: 
 The preliminaries for sequences, FR codes and Flower codes are collected in Section \ref{Section Preliminaries}. 
 The existence of Flower code and condition for universally good Flower code on sequences are investigated in Section \ref{Section Results}. 
 Section \ref{Section Conclusions} concludes the paper followed by references. The appendices contains the proofs. 
\section{Preliminaries}\label{Section Preliminaries}
In this section, we have listed the basic definition and properties for sequences, FR codes and Flower codes.
\subsection{Sequences}
A \textit{sequence} is a one dimensional array defined on a set of $q$ symbols, called alphabet. 
If the length of the array is finite then the sequence is called a \textit{finite sequence}. 
If the alphabet is $\{0,1\}$ then the sequence is called a \textit{binary sequence}. 
A finite binary sequence of length $\ell\ (\in\N)$ is denoted by $\x=(x_1\ x_2\ldots x_\ell)$, where $x_i\in\{0,1\}$ for $i=1,2,\ldots,\ell$. 
The \textit{weight} of a finite binary sequence $\x$ is the sum of all the terms of the sequence and denoted by $w_\x$. 
The \textit{weight} of the first $s$ terms of a binary sequence $\x$ is the sum of the first $s$ terms of the sequence and denoted by $w_\x(s)$. 
For example, if $\x=101101$ then $w_\x=4$ and $w_\x(3)=2$.
A sequence $\x$ of length $\ell$ is called \textit{periodic sequence}, if there exists some positive integer $\tau$ such that $\x_r=\x_{r+\tau}$, for $r=1,2,\ldots,\ell-\tau$. 
For the periodic sequence, the parameter $\tau$ is called \textit{period}.
For example, the sequence $\x=1011 1011 1011$ is a periodic sequence with the period $\tau=4$. 
For two sequences $\x$ and $\y$, the concatenation is denoted $\x\y$.
For a positive integer $s$, the sequence $\x^s$ denotes the concatenation of $s$ copies of the sequence $\x$. 
For example, $(10)^2(1011)^30^4=10 10 1011 1011 1011 0000$.

\subsection{Fractional Repetition Code}
Let a data file be broken into $M$ packets of the same length, where the symbols are defined on a field $\F_q$ with $q$ symbols. 
Using $(\theta, M)$ MDS code, the packets are encoded into $\theta$ distinct packets and denote the encoded packets by $P_j$ ($j=1,2,\ldots,\theta$).
The packet $P_j$ is replicated $\rho_j$ times and all the $\sum_{j=1}^\theta \rho_j$ replicated packets are distributed on $n$ nodes.
The nodes are denoted by $U_i$ for $i=1,2,\ldots,n$.  
For $i=1,2,\ldots,n$, let $\alpha_i$ packets be stored in node $U_i$. 
The parameter $\rho_j$ is called the replication factor of the packet $P_j$ and the parameter $\alpha_i$ is called storage capacity of node $U_i$.
Formally, one can define an FR code as following.

\begin{definition} (Fractional Repetition Code): 
A Fractional Repetition (FR) code $\mathscr{C}:(n, \theta, \alpha, \rho)$ with $n$ nodes and $\theta$ packets is a collection of $n$ subsets $U_1,\ U_2, \ldots, U_n$ of a set $\{P_j:j=1,2,\ldots\theta\}$ such that 
\begin{itemize}
    \item $|U_i|=\alpha_i$ for each $i=1,2,\ldots,n$ and 
    \item subsets contain $\rho_j$ copies of $P_j$ for each $j=1,2,\ldots,\theta$, 
\end{itemize}
 where $\rho = \max\{\rho_j:j=1,2,\ldots,\theta\}$ and $\alpha=\max\{\alpha_i:i=1,2,\ldots,n\}$. 
 Clearly, $\sum_{j=1}^{\theta}\rho_j = \sum_{i=1}^n\alpha_i$.
\label{FR definition}
\end{definition}
An example of an FR code $\mathscr{C}: (4, 5, 3, 2)$ is illustrated in the Table \ref{FR example}, where $U_1=\{P_1,P_2,P_3\}$, $U_2=\{P_1,P_4,P_5\}$, $U_3=\{P_2,P_4\}$ and $U_4=\{P_3,P_5\}$. 
For the FR code $\mathscr{C}: (4, 5, 3, 2)$, $\alpha_1=\alpha_2=3$, $\alpha_3=\alpha_4=2$ and $\rho_j=2$ for each $j=1,2,3,4,5$. 
\begin{table}[ht]
\caption{Packet Distribution for the FR code $\mathscr{C}: (4, 5, 3, 2)$.}
\centering 
\begin{tabular}{|c||c|}
\hline
\textbf{Node} &\textbf{Packet distribution}        \\[0.5ex]
\hline\hline
$U_1$& $P_1,\ P_2,\ P_3$  \\ 
\hline
$U_2$& $P_1,\ P_4,\ P_5$  \\ 
\hline
$U_3$& $P_2,\ P_4$        \\ 
\hline
$U_4$& $P_3,\ P_5$        \\ 
\hline
\end{tabular}
\label{FR example}
\end{table}

For reliability, if a node $U_i$ ($i=1,2,\ldots,n$) fails in an FR code then the node will be repaired by replacing the failed node with a new node and the packets of the new node are downloaded from some active nodes. 
The active nodes are called \textit{helper nodes} and the set of those helper nodes is called \textit{surviving set}. 
The cardinality of the surviving set ($i.e.$ the number of active nodes from which packets are downloaded during the repair) is called \textit{repair degree} $d_i$ of the node $U_i$ for $i=1,2,\ldots,n$. 
For an FR code, the maximum repair degree $d=\max\{d_i:i=1,2,\dots,n\}$. 
To repair a node of an FR code, the total number of distinct packets which are communicated, is called \textit{repair bandwidth}. 
For example, if node $U_1$ of the FR code $\mathscr{C}: (4, 5, 3, 2)$ (Table \ref{FR example}) fails then it will be replaced with new node $U_1'=\{P_1,P_2,P_3\}$, where the packets $P_1$, $P_2$ and $P_3$ of the node $U_1'$ are downloaded from the nodes $U_2$, $U_3$ and $U_4$ respectively. 
In this example, the repair degree $d_1=3$ and the repair bandwidth is $3$ unit.

For an FR code $\mathscr{C}:(n, \theta, \alpha, \rho)$, if a data collector connects any $k<n$ nodes and gets $M$ packets out of the $\theta$ packets then the parameter $k$ is called \textit{reconstruction degree}. 
Because of ($\theta,M$) MDS code, note that the stored file information can be extracted from any $M$ distinct encoded packets of the $\theta$ packets. 
If a set of any $k$ nodes has at least $M$ distinct packets then a data collector can reconstruct the complete file by connecting any $k$ nodes.
Hence, the file size $M\leq \min\left\{\left\vert\bigcup_{i\in I}U_i\right\vert:|I|=k,I\subset\{1,2,\ldots,n\}\right\}$, where at least $\min\left\{\left\vert\bigcup_{i\in I}U_i\right\vert:|I|=k,I\subset\{1,2,\ldots,n\}\right\}$ encoded distinct packets are stored in any $k$ nodes. 
Since one always looks for maximum file size that can be stored in the FR code, we assume that $M=\min\left\{\left\vert\bigcup_{i\in I}U_i\right\vert:|I|=k,I\subset\{1,2,\ldots,n\}\right\}$. 
Note, that the same FR code can be used to store data, with different $k$'s as reconstruction degrees, and different MDS codes. 
The file size $M$, which is the dimension of the chosen MDS code, depends on the value of chosen $k$ and hence we will use $M(k)$ to denote the maximum file size $i.e.$
$M(k)=\min\left\{\left\vert\bigcup_{i\in I}U_i\right\vert:|I|=k,I\subset\{1,2,\ldots,n\}\right\}$. 

Consider an FR code $\mathscr{C}:(n,\theta,\alpha,\rho)$ with $\alpha_i=\alpha$ for each $i=1,2,\ldots,n$ and $\rho_j=\rho$ for each $j=1,2,\ldots,\theta$. 
The FR code is called \textit{universally good} \cite{rr10,7118709,8309370} if
\begin{equation}
M(k)\geq k\alpha-\binom{k}{2},
\label{file size salim}
\end{equation}
for any $k\leq\alpha$, where the right hand side of the Inequality (\ref{file size salim}) is the maximum file size that can be stored using an MBR code, $i.e.$, the MBR capacity \cite{5550492}. 
Consider an FR code $\mathscr{C}:(n, \theta, \alpha, \rho)$ in which any two nodes do not share more then one packet. 
For the FR code with the reconstruction degree $k$, the maximum file size is
\begin{equation}
M(k)\geq \sum_{i=1}^k\alpha_i-\binom{k}{2},
\label{generalized condition}
\end{equation}
where $\alpha_1\leq \alpha_2\leq\ldots\leq \alpha_n$ and the node storage capacity of the node $U_i$ is $\alpha_i$ for $i=1,2,\ldots,n$.
In \cite[Theorem 4]{6763122}, the Inequality (\ref{generalized condition}) is proved for FR codes with identical replication factor, where any two nodes share at most one packet. 
The proof for the Inequality (\ref{generalized condition}), for an FR code with $U_i\cap U_j\leq1$ ($i, j\in\{1,2,\ldots,n\}$ and $i\neq j$), is similar to the proof of \cite[Theorem 4]{6763122}.

Consider two FR codes $\mathscr{C}:(n,\theta,\alpha,\rho)$ and $\mathscr{C}^*:(n^*, \theta^*, \alpha^*, \rho^*)$ such that, in the FR code $\mathscr{C}$, the packet $P_j$ is stored on the node $U_i$ if and only if, in the FR code $\mathscr{C}^*$, the packet $P^*_i$ is stored on the node $U^*_j$, where $n^*=\theta$, $\theta^*=n$, $\alpha_j^*=\rho_j$ and $\rho_i^*=\alpha_j$ for $i=1,2,\ldots,n$ and $j=1,2,\ldots,\theta$.
Essentially, the roles of packets and nodes are exchanged between $\mathscr{C}$ and $\mathscr{C}^*$. 
Hence, the FR code $\mathscr{C}^*$ is called the dual of the FR code $\mathscr{C}$.
In \cite{8277971}, the authors studied the dual FR code. 
A formal definition of dual FR code is following. 

\begin{definition}
Let $\mathscr{C}:(n, \theta, \alpha, \rho)$ be an FR code. 
A \textit{dual} $\mathscr{C}^*:(n^*, \theta^*, \alpha^*, \rho^*)$ of the FR code $\mathscr{C}$ is an FR code such that the node $U^*_j=\{f(U_i)=P^*_i: P_j\in U_i \}$, and the bijection $f$ is defined between the set of packets $\{P^*_1,P^*_2,\ldots,P^*_{\theta^*}\}$ and the set of nodes $\{U_1,U_2,\ldots,U_n\}$. 
The parameters of the dual FR code are $n^*=\theta$, $\theta^*=n$, $\alpha_j^*=\rho_j$ and $\rho_i^*=\alpha_i$, where $i=1,2,\ldots,n$ and $j=1,2,\ldots,\theta$.
\label{dual FR code}
\end{definition}
For the FR code $\mathscr{C}:(4, 5, 3, 2)$ (as given in the Table \ref{FR example}), the dual FR code $\mathscr{C}^*:(5,4,2,3)$ is $U_1^*=\{P_1^*,P_2^*\}$, $U_2^*=\{P_1^*,P_3^*\}$, $U_3^*=\{P_1^*,P_4^*\}$, $U_4^*=\{P_2^*,P_3^*\}$ and $U_5^*=\{P_2^*,P_4^*\}$.

For any FR code $\mathscr{C}$, if $U_i\cap U_j\leq1$ ($i,j\in\{1,2,\ldots,n\}$ and $i\neq j$) then, by duality, $U^*_i\cap U^*_j\leq1$ ($i,j\in\{1,2,\ldots,n^*\}$ and $i\neq j$).
Therefore, the dual of an universally good FR code is universally good. 
Hence, the following theorem gives the connection between universally good condition on FR code and its dual code.
\begin{theorem}
The dual of a universally good FR code is also universally good.
\label{U G dual}
\end{theorem}
\subsubsection{FR code constructions using sequences}
A ring construction of FR codes for $\rho=2$ is given in \cite{DBLP:journals/corr/abs-1302-3681}.
In \cite{7458383}, the ring construction is generalized for $\rho\geq2$. 
The authors constructed FR codes by dropping packets on nodes in specific order. 
For that first arrange all nodes on a circle and drop one copy of all $\theta$ packets on the nodes one by one.
The dropping of the $\theta$ distinct packets no nodes are called a cycle \cite{7458383}.
After $\rho$ number of cycles, it reises an FR code on $n$ nodes and $\theta$ packets where all the $\rho$ copies of each packet are stored in the system. 
In this construction, the dropping of packets are controlled by a finite binary sequence. 
For a positive integer $p\leq\ell$ and a finite binary sequence $\x$ of length $\ell$, if $x_p=1$ then the selected packet will be dropped on the node and if $x_p=0$ then the selected packet will not be dropped on the node. 
For example, consider $n=4$ and $\theta=6$ and the sequence $\x=1011010100100100100010100101$. 
Now, first drop the packet $P_1$ on the node $U_1$ because $x_1=1$ and then drop the packet $P_2$ on node $U_3$ since $x_2=0$ and $x_3=1$ and so on.
Note that there are two cycles for the binary sequence $\x$ for $n=4$ and $\theta=6$.
In Figure \ref{example flower figure old}, the FR code construction is illustrated for the binary sequence $\x$ with $n=4$ and $\theta=6$. 
In the example (Figure \ref{example flower figure old}), there are two cycles and packets dropped for Cycle $1$ and $2$ which are coloured by red and blue respectively. 

\tikzstyle{vertex} = [fill,shape=circle,node distance=60pt]
\tikzstyle{edge} = [fill,opacity=.5,fill opacity=.3,line cap=round, line join=round, line width=30pt]
\tikzstyle{line} = [draw, -]
\tikzstyle{cloud} = [draw, circle,fill=red!20, node distance=2cm, minimum height=2em]
\tikzstyle{block 1} = [rectangle, draw, fill=red!20, text width=6em, text centered, rounded corners, minimum height=1em, node distance=1.7cm]
\tikzstyle{block 2} = [rectangle, draw, fill=red!20, text width=6em, text centered, rounded corners, minimum height=1em, node distance=0.7cm]
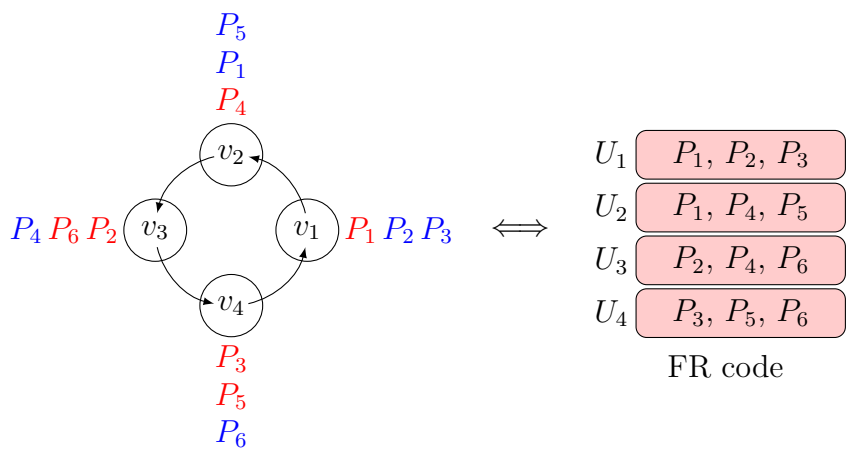
\begin{figure}
	\centering
	\begin{tikzpicture}[node distance = 2cm, auto]
	\def \n {4}
	\def \radius {1cm}
	\def \margin {13} 
	
	\foreach \s in {1,...,\n}
	{
		\node[draw, circle] at ({360/\n * (\s - 1)}:\radius) {$v_{\s}$};
		\draw[->, >=latex] ({360/\n * (\s - 1)+\margin}:\radius) 
		arc ({360/\n * (\s - 1)+\margin}:{360/\n * (\s)-\margin}:\radius);
	}
	
	\draw [red] (1.7,0) node {$P_1$};
	\draw [blue] (2.2,0) node {$P_2$};
	\draw [blue] (2.7,0) node {$P_3$};
	
	\draw [red] (0,1.7) node {$P_4$};
	\draw [blue] (0,2.2) node {$P_1$};
	\draw [blue] (0,2.7) node {$P_5$};
	
	\draw [red] (-1.7,0) node {$P_2$};
	\draw [red] (-2.2,0) node {$P_6$};
	\draw [blue] (-2.7,0) node {$P_4$};
	
	\draw [red] (0,-1.7) node {$P_3$};
	\draw [red] (0,-2.2) node {$P_5$};
	\draw [blue] (0,-2.7) node {$P_6$};
	
	\draw (3.8,0) node {$\Longleftrightarrow$};
	\node[label=right:\(\)]  (u1) at (5,1) {};
	\node [block 1, right of=u1] (u3) {$P_1$, $P_2$, $P_3$};
	\node [block 2, below of=u3] (u4) {$P_1$, $P_4$, $P_5$};
	\node [block 2, below of=u4] (u5) {$P_2$, $P_4$, $P_6$};
	\node [block 2, below of=u5] (u6) {$P_3$, $P_5$, $P_6$};
	\draw (5,1) node {$U_1$};
	\draw (5,0.3) node {$U_2$};
	\draw (5,-0.4) node {$U_3$};
	\draw (5,-1.1) node {$U_4$};
	\draw (6.5,-1.8) node {FR code };
	\end{tikzpicture}
	\caption{
		FR code construction for $n=4$, $\theta=6$, and $\x=1011010100100100100010100101$.
	}
	\label{example flower figure old}
\end{figure}
\section{Results}\label{Section Results}
In this section, a general approach for the construction of an FR code using binary sequences is discussed. 
The properties and universally good constraints for Flower codes are also discussed. 

\subsection{Flower Code: A General Approach}
Motivated by the ring construction of FR codes as given in \cite{7458383,DBLP:journals/corr/abs-1302-3681}, an FR code $\mathscr{C}: (n, \theta, \alpha , \rho)$ with different replication factor can be constructed by dropping $\theta$ packets on $n$ nodes one by one such that an arbitrary packet $P_j$ ($j=1,2,\ldots,\theta$) is replicated $\rho_j$ times in the system and the size of node $U_i$ ($i=1,2,\ldots,n$) is $\alpha_i$. 
To construct the FR code, first place all $n$ nodes on a circle and then drop all replicated packets on those nodes one by one. 
One has to place the packets till all the replicated packets are consumed. 
This will give rise to an FR code since all replicas of each packet are in the system. 
For example, consider $n=4$ distinct nodes $U_1$, $U_2$, $U_3$ and $U_4$, and $\theta=4$ distinct packets $P_1$, $P_2$, $P_3$ and $P_4$. 
The replication factor of those packets are $\rho_1=3,\rho_2=2,\rho_3=2$ and $\rho_4=2$. 
Let $\{P_2,P_1,P_1,P_3,P_4,P_2,P_4,P_1,P_3\}$ be a collection of all the replicas of all four packets. 
A Flower code can be constructed by selecting packets one by one from the collection and dropping it on nodes one by one, where all the $4$ nodes are placed on a circle (see Figure \ref{example figure}).
\begin{definition}
A system on ($n,k$)-DSS is called a Flower code in which packets selected from a collection are distributed among $n$ nodes arranged on a circle, where the collection contains all the replicas of each packet.
\end{definition}

\tikzstyle{vertex} = [fill,shape=circle,node distance=60pt]
\tikzstyle{edge} = [fill,opacity=.5,fill opacity=.3,line cap=round, line join=round, line width=30pt]
\tikzstyle{line} = [draw, -]
\tikzstyle{cloud} = [draw, circle,fill=red!20, node distance=2cm, minimum height=2em]
\tikzstyle{block 1} = [rectangle, draw, fill=red!20, text width=6em, text centered, rounded corners, minimum height=1em, node distance=1.7cm]
\tikzstyle{block 2} = [rectangle, draw, fill=red!20, text width=6em, text centered, rounded corners, minimum height=1em, node distance=0.7cm]
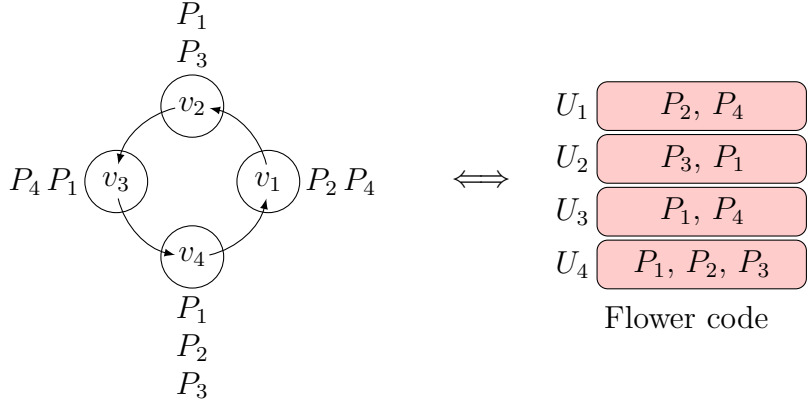
\begin{figure}
	\centering
	\begin{tikzpicture}[node distance = 2cm, auto]
	\def \n {4}
	\def \radius {1cm}
	\def \margin {13} 

	
	
	
	\foreach \s in {1,...,\n}
	{
		\node[draw, circle] at ({360/\n * (\s - 1)}:\radius) {$v_{\s}$};
		\draw[->, >=latex] ({360/\n * (\s - 1)+\margin}:\radius) 
		arc ({360/\n * (\s - 1)+\margin}:{360/\n * (\s)-\margin}:\radius);
	}
	
	\draw (1.7,0) node {$P_2$};
	\draw (2.2,0) node {$P_4$};
	\draw (0,1.7) node {$P_3$};
	\draw (0,2.2) node {$P_1$};
	\draw (-1.7,0) node {$P_1$};
	\draw (-2.2,0) node {$P_4$};
	\draw (0,-1.7) node {$P_1$};
	\draw (0,-2.2) node {$P_2$};
	\draw (0,-2.7) node {$P_3$};
	
	
	
	\draw (3.8,0) node {$\Longleftrightarrow$};
	\node[label=right:\(\)]  (u1) at (5,1) {};
	\node [block 1, right of=u1] (u3) {$P_2$, $P_4$};
	\node [block 2, below of=u3] (u4) {$P_3$, $P_1$};
	\node [block 2, below of=u4] (u5) {$P_1$, $P_4$};
	\node [block 2, below of=u5] (u6) {$P_1$, $P_2$, $P_3$};
	\draw (5,1) node {$U_1$};
	\draw (5,0.3) node {$U_2$};
	\draw (5,-0.4) node {$U_3$};
	\draw (5,-1.1) node {$U_4$};
	\draw (6.5,-1.8) node {Flower code };
	\end{tikzpicture}
	\caption{
		An example of a Flower code.
	}
	\label{example figure}
\end{figure}
To construct a Flower code, one can select the copies of coded packets randomly from the collection of replicated packets and he can place it one by one on nodes using some model such as the balls-and-bins model. 
For the system, one can calculate the probability of receiving complete file by user and the probability of data recovery when a node has failed.
On the other hand, the selection of packets from the collection and dropping the selected packet on nodes can be done using finite binary characteristic sequences.
More precisely, for finite binary characteristic sequences $\y$ (of length $t$) and $\x$ (of length $\ell$), if $y_r=1$ ($r\in\{1,2,\ldots,t\}$) then the associated packet is selected to be dropped on a node, and if $x_m=1$ ($m\in\{1,2,\ldots,\ell\}$) then the selected packet is dropped on the associated node. 
For a Flower code, both sequences are called the \textit{packet selection sequence} and the \textit{packet dropping sequence} respectively. 
For the terms $y_j=1$ and $x_m=1$, one can select and drop a packet on a node in various methods. 
In this paper, we consider the following:
\begin{itemize}
    \item if the term $y_r=1$ then the packet for which the packet-index is congruent to $r\ (mod\ \theta)$ is selected to be dropped on a node, and
    \item if the term $x_m=1$ then the selected packet is dropped on the node for which the node-index is congruent to $m\ (mod\ n)$. 
\end{itemize}  
For example, Flower code as considered in the Figure \ref{example figure} can be constructed using the two binary sequences $\y=0100100010110101101$ and $\x=101101111101$. 
In the rest of the paper, for a flower code with $n$ nodes and $\theta$ packets, the packet selection sequence of length $t$ is denoted by $\y$ and the packet dropping sequence of length $\ell$ is denoted by $\x$.

The necessary conditions for the existence of a Flower code for any two binary sequences are following.
\begin{lemma}
For $n$ nodes, $\theta$ packets, and binary sequences $\y$ and $\x$, a Flower code exists, if 
\begin{enumerate}
    \item\label{1} the weight of both sequences are same $i.e.$ $w_\x=w_\y$, 
    \item\label{2} $w_\x\geq\theta$ and $w_\y\geq n$.
\end{enumerate}
\label{Gen Flower existence}
\end{lemma}
From the construction of Flower code, one can observe the following Fact.
\begin{fact}
For some positive integers $m\leq\ell$ and $r\leq t$, consider a Flower code with $n$ nodes, $\theta$ packets and two characteristic binary sequences $\y$ and $\x$.
For $i\in\{1,2,\ldots,n\}$ and $j\in\{1,2,\ldots,\theta\}$, the packet $P_j$ is dropped on the node $U_i$ if and only if $x_m=1$, $y_r=1$ and $w_\x(m)=w_\y(r)$, where $i-m\equiv0\ (mod\ n)$ and $j-r\equiv0\ (mod\ \theta)$.
\label{node packets indexes}
\end{fact}

The parameters of a Flower code are calculated in the following theorem.
\begin{theorem}
The parameters of a Flower code with $n$ nodes, $\theta$ packets and characteristic binary sequences $\y$ and $\x$ are $\rho_j= \sum_{p=0}^{\left\lfloor \frac{t-j}{\theta}\right\rfloor}y_{p\theta+j}$ ($j=1,2,\ldots,\theta$) and $\alpha_i= \sum_{s=0}^{\left\lfloor \frac{\ell-i}{n}\right\rfloor}x_{sn+i}$ ($i=1,2,\ldots,n$).
\label{Gen Flower parameters}
\end{theorem}
From the Theorem \ref{Gen Flower parameters}, one can observe the following propositions.
\begin{proposition}
For each $p=0,1,\ldots,\left\lfloor \frac{t-j}{\theta}\right\rfloor$, $y_{p\theta+j}=0$ if and only if the packet $P_j$ does not have any copy in the system, where some $j\in\{1,2,\ldots,\theta\}$.
\end{proposition}
\begin{proposition}
	For each $s=0,1,\ldots,\left\lfloor \frac{\ell-i}{n}\right\rfloor$, $x_{sn+n}=0$ if and only if the node $U_i$ does not contain any packet, where some $i\in\{1,2,\ldots,n\}$.
\end{proposition}
The following remark gives the constraint on binary sequences for the distribution of multiple copies of a particular packet on a specific node.
\begin{remark}
For $n$ nodes and $\theta$ packets, consider a Flower code with binary sequences $\x$ and $\y$.
For positive integers $m_1<m_2\leq\ell$ and $r_1<r_2\leq t$, let $x_{m_1}=x_{m_2}=1$ and $y_{r_1}=y_{r_2}=1$ such that $w_\x(m_1)=w_\y(r_1)$ and $w_\x(m_2)=w_\y(r_2)$. 
If $m_1-m_2\equiv0\ (mod\ n)$ and $r_1-r_2\equiv0\ (mod\ \theta)$ then the two replicas of a particular packet are dropped on a specific node.  
\end{remark}
From the construction of FR code, it is easy to observe that Flower code exists for each FR code with any set of parameters.
\begin{lemma}
For any FR code, a Flower code exists.
\label{Flower exist}
\end{lemma}
For given $i\in\{1,2,\ldots,n\}$, $j\in\{1,2,\ldots,\theta\}$ and a Flower code with $n$ nodes, $\theta$ packets and binary sequences $\x$ and $\y$, let $A_i(j)$ be the number of copies of the packet $P_j$ which are stored in the node $U_i$ in the Flower code.
For example, in the Flower code (see figure \ref{example figure}) one copy of the packet $P_2$ is stored in the node $U_1$ so, $A_1(2)=1$. similarly, $A_1(3)=0$, because packet $P_3$ is not stored in node $U_1$.
For a Flower code, from the Fact \ref{node packets indexes} and the Theorem \ref{Gen Flower parameters}, one can observe the following proposition.
\begin{proposition}
Consider a Flower code with $n$ nodes, $\theta$ packets and binary sequences $\x$ and $\y$.
For given $i\in\{1,2,\ldots,n\}$ and $j\in\{1,2,\ldots,\theta\}$, if $I_{i,j}=\{r:w_\y(r)=w_\x(m),r-j\equiv0\ (mod\ \theta),m-i\equiv0\ (mod\ n)\}$ then $A_i(j)=\sum_{r\in I_{i,j}}y_r$.
\label{Aij proposition}
\end{proposition}
The following Lemma gives the the necessary and sufficient condition for universally good Flower code $i.e.$ universally good FR code.
\begin{lemma}
Consider a Flower code with $n$ nodes, $\theta$ packets and binary sequences $\x$ and $\y$.
The Flower code is universally good if and only if, for all integers $1\leq i\neq p\leq \ell$, 
\begin{equation}
\sum_{j=1}^\theta A_i(j)A_p(j)\leq 1,
\label{U G Condition Flower code equation}
\end{equation}
 where, $A_i(j)$ is given in the Proposition \ref{Aij proposition}.
	\label{U G Condition Flower code lamma}
\end{lemma}
For the Flower code (see the Figure \ref{example figure}) with $n=\theta=4$, and two binary sequences $\y=0100100010110101101$ and $\x=101101111101$, $A_1=0101$, $A_2=1010$, $A_3=1001$ and $A_4=1110$.
Note that the Flower code is not universally good because $\sum_{j=1}^4A_2(j)A_4(j)=2\nleq 1$ $i.e.$ node $U_2$ and node $U_4$ share $2$ distinct packets. 
\begin{remark}
For any integers $1\leq i<p<s\leq n$, if a binary sequence $\x$ satisfies $\sum_{j=1}^\theta A_i(j)A_p(j)\leq1$ and $\sum_{j=1}^\theta A_i(j)A_p(j)A_s(j)=0$ then any two nodes share no more than one packet. 
Therefore, the Flower code is universally good.
\label{UG rho}
\end{remark}
\begin{remark}
From the Remark \ref{UG rho}, it is easy to observe that if a binary sequence $\x$ with the entry sum $2\theta$ satisfies $\sum_{j=1}^\theta A_i(j)A_p(j)\leq1$ for any $i,p\in\{1,2,\ldots,n\}$ and $i\neq p$ then the Flower is universally good. 
\label{UG rho 2}
\end{remark}

For a Flower code, the maximum file size is calculated in the following Lemma.
\begin{lemma}
For a universally good Flower code with $n$ nodes, $\theta$ packets and reconstruction degree $k$, the maximum file size
\begin{equation}
    M(k)=\max\left\{\sum_{i\in I}\sum_{j=1}^\theta A_i(j)-\sum_{i,p\in I;i<p}\sum_{j=1}^\theta A_i(j)A_p(j):|I|=k, I\subset\{1,2,\ldots,n\}\right\}.
    \label{maximum file size Flower theorem}
\end{equation}
\label{maximum file size Flower lemma}
\end{lemma}
The following subsection discusses a general construction of the two characteristic binary sequences for a Flower code.
\subsubsection{Construction of universally good Flower code}
For given positive integers $n$, $\theta$ and $z$, Algorithm \ref{U G Sequence generation} ensures two binary sequences $\x$ and $\y$ for which the Flower code is universally good. 
For a Flower code the value $z$ is the count of copies of packets which are placed in the system.
In the Algorithm \ref{U G Sequence generation}, for each $j=1,2,\ldots,\theta$ and $i=1,2,\ldots,n$, first initialize $A_i(j)$ with $0$.
By permuting node-indices and permuting packet-indices, for any Flower code, one can always find a Flower code in which the packet $P_1$ is dropped on the node $U_1$, where both the Flower codes have same properties. 
So, the binary sequences are initialized with $\x=1$ and $\y=1$ which ensures the packet $P_1$ is dropped on the node $U_1$.
In each while loop, the sequences $\x$ and $\y$ which satisfy the Inequality (\ref{U G Condition Flower code equation}) are concatenated with $0^{a-1}1$ and $0^{b-1}1$. 
For the loop, the positive integers $a$ and $b$ are selected such that those concatenated sequences $\x 0^{a-1}1$ and $\y 0^{b-1}1$ satisfy the Inequality (\ref{U G Condition Flower code equation}).
Note that $w_{\x0^{a-1}1}=w_\x+1$ and $w_{\y0^{b-1}1}=w_\y+1$ and hence, by induction with initial condition, $x_1=y_1=1$, $w_{\x0^{a-1}1}=w_{\y0^{b-1}1}$. 
In the Algorithm \ref{U G Sequence generation}, the positive integers $a$ and $b$ can be distinct for different while loop as well as within a loop.
One can introduce randomness by selecting $a$ and $b$ as random integers. 
The algorithm will terminate, if $w_\x=w_\y>z$.
The condition on while loop ensures the the weight of each binary sequence is $z$. 

An Example of the construction of sequences is illustrated in Table \ref{Algorithm example table} for $n=3$ and $\theta=3$.
From the algorithm, $\x=11001100101$ and $\y=100101010011$.
The Flower code is $U_1=\{P_1,P_2\}$, $U_2=\{P_1,P_3\}$, $U_3=\{P_2,P_3\}$. 
Note that the Flower code is universally good. 
\begin{table}[ht]
\caption{For $n=\theta=4$, construction of sequence $\x$ and $\y$ using the Algorithm \ref{U G Sequence generation}.}
\centering 
\begin{tabular}{|c|c|l|c|l|c|c|}
\hline
Sr. No. & $a$ & \hspace{1.4cm}$\x$      & $b$ & \hspace{1.4cm} $\y$     & Inequality \ref{U G Condition Flower code equation} & $A_i(j)$   \\ [0.5ex]
\hline
1.      & -   & $1$                     & -   & $1$                     & satisfied                                           & $A_1(1)=1$ \\ 
\hline
2.      & $1$ & $10^01$                 & $3$ & $10^21$                 & satisfied                                           & $A_2(1)=1$ \\ 
\hline
3.      & $3$ & $10^010^21$             & $2$ & $10^210^11$             & satisfied                                           & $A_2(3)=1$ \\ 
\hline
4.      & $1$ & $10^010^210^01$         & $2$ & $10^210^110^11$         & satisfied                                           & $A_3(2)=1$ \\ 
\hline
5.      & $3$ & $10^010^210^010^21$     & $3$ & $10^210^110^110^21$     & satisfied                                           & $A_1(2)=1$ \\ 
\hline
6.      & $2$ & $10^010^210^010^210^11$ & $1$ & $10^210^110^110^210^01$ & satisfied                                           & $A_3(3)=1$ \\
\hline
\end{tabular}
\label{Algorithm example table}
\end{table}

\begin{algorithm} 
	\caption{Construction of sequences $\x$ and $\y$ for a universally good Flower code.} 
	\begin{algorithmic} 
		\REQUIRE $n,\theta \geq 1$ and an integer $z\geq\max\{n,\theta\}$. 
		\ENSURE The finite binary sequences $\x$ and $\y$.
		\STATE  Initialize: $A_i(j)=0$ for each $j=1,2,\ldots,\theta$ and $i=1,2,\ldots,n$.
		\STATE  Initialize: $\x=x_1=1$, $\y=y_1=1$ and update $A_1(1)=1$. 
		\WHILE{$w_\x(m) < z$ for $m>1$}
		\STATE For positive integers $a\ (\leq n)$ and $b\ (\leq \theta)$
		\IF{$A_p(j)+\sum_{s=1,s\neq j}^\theta A_i(s)A_p(s)\leq 1$ for $p=1,2,\ldots,n$ and $i\neq p$, where $i-(m+a)\equiv 0\ (mod\ n)$ and $j-(r+b)\equiv 0\ (mod\ \theta)$}
		\STATE Update $\x$ by $\x 0^{a-1}1$ and $\y$ by $\y 0^{b-1}1$
		\STATE Update $A_i(j)$ by $A_i(j)+1$
		\STATE Update $m$ by $m+a$, and $r$ by $r+b$
		\ENDIF
		\ENDWHILE
		\STATE Update $\ell$ by $m$, and $t$ by $r$
	\end{algorithmic}
	\label{U G Sequence generation} 
\end{algorithm}
For the binary sequences generated by the Algorithm \ref{U G Sequence generation}, the following Lemma ensures the universally good condition of the Flower code.
\begin{lemma}
For positive integers $n$ and $\theta$, if the binary sequences $\x$ and $\y$ are generated by the Algorithm \ref{U G Sequence generation} then the Flower code is universally good.
\label{Algo lemma proof}
\end{lemma}

\subsubsection{Dual of a Flower code}
The Flower code which corresponds to the dual FR code is called a dual Flower code. 
The following theorem gives the connection between a Flower code and its dual.
\begin{theorem}
Consider a Flower code with $n$ nodes, $\theta$ packets and two binary sequences $\y$ and $\x$ (the packet selection sequence and the packet dropping sequence respectively).
A Flower code with $n^*$ nodes, $\theta^*$ packets and two binary sequences $\y^*$ and $\x^*$ (the packet selection sequence and the packet dropping sequence respectively) is the dual code if and only if $n^*=\theta$, $\theta^*=n$, $\x^*=\y$ and $\y^*=\x$.
\label{dual flower code parameters}
\end{theorem}
From the Theorem \ref{Gen Flower parameters} and \ref{dual flower code parameters}, one can easily observe the following proposition. 
\begin{proposition}
For a Flower code with $n$ nodes, $\theta$ packets and characteristic binary sequences $\y$ and $\x$, the parameters of dual Flower code are $\alpha^*_j= \sum_{p=0}^{\left\lfloor \frac{t-j}{\theta}\right\rfloor}y_{p\theta+j}$ ($j=1,2,\ldots,\theta$) and $\rho^*_i= \sum_{s=0}^{\left\lfloor \frac{\ell-i}{n}\right\rfloor}x_{sn+i}$ ($i=1,2,\ldots,n$).
\end{proposition}
Now, we explore some specific cases of Flower code such as Flower code with packet selection sequence $\y=1^t$ and Flower code with packet dropping sequence $\x=1^\ell$ in following two subsections.
\subsubsection{Flower codes with all one packet selection sequences}
In the subsection, we discuss the properties (such as parameters, universally good condition and dual property) of Flower codes with all one packet selection sequence.
The Flower code with uniform replication factor can also be constructed by taking packet selection sequence $\y=1^{\rho\theta}$. 
For $n$ nodes, $\theta$ packets, and binary sequences $\x$ and $\y=1^t$, from the Lemma \ref{Gen Flower existence}, the Flower code exists if $t=w_\x$.
From the Fact \ref{node packets indexes}, in the Flower code, packet $P_j$ is stored in node $U_i$ if and only if $x_m=1$, $w_\x(m)-j\equiv0\ (mod\ \theta)$ and $m-i\equiv0\ (mod\ n)$.
The parameters of such Flower codes are calculated in the following Theorem.
\begin{theorem}
	For a Flower code with $n$ nodes and $\theta$ packets and a binary sequence $\x$ of length $\ell$, the packet replication factor of the packet $P_j$ ($j=1,2,\ldots,\theta$) is 
	\begin{equation*}
	\rho_j = \left\{ 
	\begin{array}{ll}
	\frac{w_\x}{\theta}  & :\mbox{if } \eta = 0,   \\
	1+\left\lfloor \frac{w_\x}{\theta} \right\rfloor  & :\mbox{if }j=1,2,\ldots,\eta \mbox{ and } \eta\neq 0,    \\
	\left\lfloor \frac{w_\x}{\theta} \right\rfloor  & :\mbox{if }j=\eta+1,\ldots,\theta\mbox{ and } \eta\neq 0; 
	\end{array}
	\right. \\
	\end{equation*}
	where $\eta=w_\x-\theta\left\lfloor \frac{w_\x}{\theta}\right\rfloor$.
	\label{Flower parameter}
\end{theorem}
Using the Theorem \ref{Flower parameter}, one can observe the following remark.
\begin{remark}
For $n$ nodes and $\theta$ packets, consider a Flower code with binary sequences $\x$ (length $\ell$) and $\y=1^{w_\x}$.
For two distinct integers $m_1\ (\leq\ell)$ and $m_2\ (\leq\ell)$, let $x_{m_1}=x_{m_2}=1$. 
If $m_1-m_2\equiv0\ (mod\ n)$ and $w_\x(m_1)-w_\x(m_2)\equiv0\ (mod\ \theta)$ then the two replicas of the packet $P_j$ are dropped on a node $U_i$, where $w_\x(m_1)-j\equiv0\ (mod\ \theta)$ and $i-m_1\equiv0\ (mod\ n)$.  
\end{remark}
Using the Definition of a Flower code and the Theorem \ref{Flower parameter}, one can easily prove the following proposition.
\begin{proposition}
For a Flower code with $n$ nodes, $\theta$ packets, and binary sequences $\x$ (length $\ell$) and $\y=1^{w_\x}$, the replication factor $\rho_j\in\{\rho-1,\rho\}$ ($j=1,2,\ldots,\theta$).
\end{proposition}
\begin{lemma}
A Flower code with $n$ nodes, $\theta$ packets, and binary sequences $\x$ (length $\ell$) and $\y=1^{w_\x}$, is universally good, if and only if $\sum_{j=1}^\theta A_i(j)A_p(j)\leq 1$ for each $i,p=1,2,\ldots,n$ and $i\neq p$. 
\end{lemma}
\begin{proof}
The proof follows the Theorem \ref{U G Condition Flower code equation} for the packet selection sequence $\y=1^{w_\x}$.
\end{proof}

Using the Remark \ref{UG rho 2}, one can easily prove the following two Lemmas.
\begin{lemma}
For given integers $n$, $\theta$ and a periodic binary sequence $\x$ with the period $\tau< n<\theta$, if $\ell=\left\lceil\frac{2\theta\tau}{w_\x(\tau)}\right\rceil$ then the Flower code with two characteristic binary sequences $\y=1^{\rho\theta}$ and $\x$ will be universally good with $\rho=2$, where $gcd(n,\theta)=1$ and $gcd(n,\tau)=1$.   
\label{rho 2 example}
\end{lemma}
For example, $\tau=3$, $n=4$, $\theta=5$, the sequence $\x=110110110110110$ will give the universally good Flower code, where $\alpha_1=\alpha_2=3$ and $\alpha_3=\alpha_4=2$ and $\rho_j=2$ ($j=1,2,3,4,5$). 
Note that the sequence $\x$ is the Fibonacci sequence with modulo $2$.
\begin{proposition}
For $n$, $\theta$, and the a binary sequence $\x$, the Flower code is a universally good, where the sequence is generated from the Algorithm \ref{U G Sequence generation} with $b=1$. 
\end{proposition}

\subsubsection{Flower codes with all one packet dropping sequences}
In the subsection, we discuss the properties (such as parameters, universally good condition and dual property) of Flower codes with all one packet dropping sequence.
The Flower code with uniform node storage capacity can also be constructed by taking packet dropping sequence $\x=1^{\rho\theta}$. 
For $n$ nodes, $\theta$ packets, and binary sequences $\x=1^\ell$ and $\y$, from the Lemma \ref{Gen Flower existence}, the Flower code exists if $\ell=w_\y$.
From the Fact \ref{node packets indexes}, in the Flower code, packet $P_j$ is stored in node $U_i$ if and only if $y_r=1$, $r-j\equiv0\ (mod\ \theta)$ and $w_\y(r)-i\equiv0\ (mod\ n)$.
The parameters of such Flower codes are calculated in the following Theorem.

\begin{theorem}
For a Flower code with $n$ nodes, $\theta$ packets and a binary sequence $\y$ of length $t$ the node storage capacity of the node $U_i$ ($i=1,2,\ldots,n$) is 
	\begin{equation*}
	\alpha_i = \left\{ 
	\begin{array}{ll}
	\frac{w_\y}{n}  & :\mbox{if } \eta = 0,   \\
	1+\left\lfloor \frac{w_\y}{n} \right\rfloor  & :\mbox{if }i=1,2,\ldots,\eta \mbox{ and } \eta\neq 0,    \\
	\left\lfloor \frac{w_\y}{n} \right\rfloor  & :\mbox{if }i=\eta+1,\ldots,\theta\mbox{ and } \eta\neq 0; 
	\end{array}
	\right. \\
	\end{equation*}
	where $\eta=w_\y-\theta\left\lfloor \frac{s}{n}\right\rfloor$.
	\label{Flower parameter storage constant}
\end{theorem}
Using the Theorem \ref{Flower parameter storage constant}, one can observe the following two remarks.
\begin{remark}
For $n$ nodes and $\theta$ packets, consider a Flower code with a binary sequence $\x=1^{w_\y}$ and $\y$ (length $t$).
For two distinct integers $r_1\ (\leq t)$ and $r_2\ (\leq t)$, let $y_{r_1}=y_{r_2}=1$. 
If $r_1-r_2\equiv0\ (mod\ \theta)$ and $w_\y(r_1)-w_\y(r_2)\equiv0\ (mod\ n)$ then the two replicas of the packet $P_j$ are dropped on a node $U_i$, where $r_1-j\equiv0\ (mod\ \theta)$ and $w_\y(r_1)-i\equiv0\ (mod\ n)$.    
\end{remark}
Using the Definition of a Flower code and the Theorem \ref{Flower parameter storage constant}, one can easily prove the following proposition.
\begin{proposition}
For a Flower code with $n$ nodes, $\theta$ packets, and binary sequences $\x=1^{w_\y}$ and $\y$ (length $t$), the node storage capacity $\alpha_i\in\{\alpha-1,\alpha\}$ ($i=1,2,\ldots,n$).
\end{proposition}
\begin{lemma}
A Flower code with $n$ nodes, $\theta$ packets, and binary sequences $\x=1^{w_\y}$ and $\y$ (length $t$), is universally good, if and only if $\sum_{j=1}^\theta A_i(j)A_p(j)\leq 1$ for each $i,p=1,2,\ldots,n$ and $i\neq p$.
\label{ha ha ha} 
\end{lemma}
\begin{proposition}
For $n$, $\theta$, and the a binary sequence $\y$, the Flower code is universally good, when the sequence is generated from the Algorithm \ref{U G Sequence generation} with $a=1$. 
\end{proposition}

The proofs of theorems and lemmas are given in the Appendices.

\section{Conclusions}\label{Section Conclusions}
In this paper, we have calculated the bound for the universally good FR code using sequences.
Universally good FR codes are constructed using some families of binary sequences of finite length. 
It would be interesting to study some more bounds on FR codes using sequences.


\bibliographystyle{plain} 
\bibliography{cloud} 

\newpage
\appendix
\addcontentsline{toc}{section}{appendices}
\section{Appendices} 
\subsection{Proofs of Theorems}

\subsubsection{Proof of Theorem \ref{U G dual}}
Indeed, if an FR code is universally good then, in the dual, two packets $P^*_i$ ($1\leq i\leq \theta^*$) and $P^*_j$ ($1\leq j\leq \theta^*$) cannot be stored in two distinct nodes $U^*_r$ ($1\leq r\leq n^*$) and $U^*_s$ ($1\leq s\leq n^*$), as then, packets $P_r$ and $P_s$ are stored in both nodes $U_i$ and $U_j$ which contradicts $|U_i\cap U_j|\leq 1$.

\subsubsection{Proof of Theorem \ref{Gen Flower parameters}}
From the Fact \ref{node packets indexes}, the Theorem follows the fact that a packet $P_j$ is selected to be drop on some node if and only if $y_{p\theta+j}=1$ for some $p\in\{0,1,\ldots,\left\lfloor \frac{t-j}{\theta}\right\rfloor\}$. 
Similarly, a selected packet is dropped on a node $U_i$ if and only if $x_{sn+i}=1$ for some $s\in\{0,1,\ldots,\left\lfloor \frac{\ell-i}{n}\right\rfloor\}$.

\subsubsection{Proof of Theorem \ref{maximum file size Flower lemma}}
For a subset $I\subset\{1,2,\ldots,n\}$, the total number of packets stored in nodes (indexed from $I$) is $\sum_{i\in I}\sum_{j=1}^\theta A_i(j)$. 
The total number of common packets shared by those nodes is $\sum_{i,p\in I;i<p}\sum_{j=1}^\theta A_i(j)A_p(j)$. 
For universally good Flower code, three or more nodes do not share any packet. 
Hence, by inclusion exclusion principle, for the reconstruction degree $k=|I|$, the maximum file size is given by the Equation \ref{maximum file size Flower theorem}.

\subsubsection{Proof of Theorem \ref{dual flower code parameters}}
The proof follows the Definition \ref{dual FR code} and the Fact \ref{node packets indexes}.

\subsubsection{Proof of Theorem \ref{Flower parameter}} 
From the division algorithm, $w_\x=\theta\left\lfloor \frac{w_\x}{\theta}\right\rfloor+\eta$ for some $0\leq\eta<\theta$. 
So, $\frac{w_\x-\eta}{\theta}$ is a positive integer $\left\lfloor \frac{w_\x}{\theta}\right\rfloor$.
Now, from the Theorem \ref{Gen Flower parameters}, for $\y=1^{w_\x}$, 
\begin{equation}
\begin{split}
\rho_j=\sum_{p=0}^{\left\lfloor \frac{w_\x-j}{\theta}\right\rfloor}y_{p\theta+j}=\left\lfloor \frac{w_\x-j}{\theta}\right\rfloor & +1
=\left\lfloor \frac{w_\x-\eta}{\theta}+\frac{\eta-j}{\theta}\right\rfloor+1 \\
& = \left\lfloor \frac{w_\x-\eta}{\theta}\right\rfloor+\left\lfloor\frac{\eta-j}{\theta}\right\rfloor+1 = \left\lfloor \frac{w_\x}{\theta}\right\rfloor+\left\lfloor\frac{\eta-j}{\theta}\right\rfloor+1
\end{split}
\label{floor Gen Flower}
\end{equation}
Now, there are two cases.
\begin{itemize}
    \item Case $1$: if $0< j\leq\eta<\theta$ then $0<\frac{\eta-j}{\theta}<1$. 
    Therefore, $\left\lfloor\frac{\eta-j}{\theta}\right\rfloor=0$.
    \item Case $2$: if $0\leq\eta\leq j<\theta$ then $-1<\frac{\eta-j}{\theta}<0$. 
    So, $\left\lfloor\frac{\eta-j}{\theta}\right\rfloor=-1$.
\end{itemize}
From the Equation \ref{floor Gen Flower}, and Case $1$ and $2$, the result holds.

\subsubsection{Proof of Theorem \ref{Flower parameter storage constant}} 
The proof follows the dual of the Theorem \ref{Flower parameter}.

\subsection{Proofs of lemmas}

\subsubsection{Proof of Lemma \ref{Gen Flower existence}}
The proof of the part \ref{1}: For some $r\leq t$, if $y_r=1$ only then a packet is selected to be dropped on a node so, the weight of the binary sequence $\y$ is the total number of replicated packets which are stored in the system. 
On the other hand, for some $m\leq \ell$, if $x_m=1$ only then the selected packet is dropped on a node, so the weight of the binary sequence $\x$ is the total storage capacity of all nodes. 
But the packet replica sum and the storage capacity sum are always same for any system which follows the proof of the part \ref{1}. 
Therefore, the Flower code exists, if the weight of both the sequences are same. 
The proof of the part \ref{2} follows the fact that the total number of replicated packets which are stored in the system can not be less than $\theta$ and the total storage capacity of all nodes can not be less than $n$.

\subsubsection{Proof of Lemma \ref{Flower exist}}
For an FR code with $n$ nodes and $\theta$ packets, one can always find the collection $S=\{(i,j):\mbox{ Packet } P_j \mbox{ is stored in node } U_i,\  i\in\{1,2,\ldots,n\},\ j\in\{1,2,\ldots,\theta\}\}$ where $|S|=s$. 
Arrange all the pairs of the collection $S$ in a specific order and then extract the sequences of $i$'s and $j$'s. 
Denote the sequences by $\textbf{a}=(a_1\ a_2\ldots a_s)$ and $\textbf{b}=(b_1\ b_2\ldots b_s)$ respectively. 
Note that the sequence $\textbf{a}$ is defined on $\{1,2,\ldots,n\}$ and the sequence $\textbf{b}$ is defined on $\{1,2,\ldots,\theta\}$.
For any positive integer $z\leq|S|$, the packet $P_{b_z}$ is stored on the node $U_{a_z}$ in the FR code.
For the sequences $\textbf{a}$ and $\textbf{b}$, construct binary sequences $\x=0^{a_1-1}10^{a_2-a_1-1}1\ldots10^{a_s-a_{s-1}-1}1$ and $\y=0^{b_1-1}10^{b_2-b_1-1}1\ldots10^{b_s-b_{s-1}-1}1$, where subtractions for the sequences $\x$ and $\y$ are modulo $n$ and modulo $\theta$ respectively. 
The binary sequences $\x$ and $\y$ are finite so assume the length of $\x$ and $\y$ are $\ell$ and $t$ respectively. 
For some integers $m\leq\ell$ and $r\leq t$, let $w_\x(m)=w_\y(r)$, $x_m=1$ and $y_r=1$. 
For the integers $m$ and $r$, one can find a positive integer $z\leq s$ such that $m=\left((a_1-1)+1+(a_2-a_1-1)+1+\ldots+1+(a_z-a_{z-1}-1)+1\right)\ (mod\ n)\equiv a_z$ and $r=\left((b_1-1)+1+(b_2-b_1-1)+1+\ldots+1+(b_z-b_{z-1}-1)+1\right)\ (mod\ \theta)\equiv b_z$. 
Hence, $m-a_z\equiv0\ (mod\ n)$ and $m-b_z\equiv0\ (mod\ \theta)$.
From the Fact \ref{node packets indexes}, the packet $P_{b_z}$ is stored in node $U_{a_z}$ in the Flower code.

\subsubsection{Proof of Lemma \ref{U G Condition Flower code lamma}}
From the Fact \ref{node packets indexes}, for positive integers $i\in\{1,2,\ldots,n\}$ and $j\in\{1,2,\ldots,\theta\}$, $A_i(j)=1$ if and only if the packet $P_j$ is dropped on the node $U_i$. 
For any $p\in\{1,2,\ldots,n\}$ such that $p\neq i$, the value $\sum_{j=1}^\theta A_i(j)A_p(j)$ is the total number of distinct packets common in both the nodes $U_i$ and $U_p$.
Hence, the Flower code is universally good, if $\sum_{j=1}^\theta A_i(j)A_p(j)\leq1$ for $1\leq i\neq p\leq n$.

\subsubsection{Proof of Lemma \ref{Algo lemma proof}}
For positive integers $a$ and $b$ in a loop of the Algorithm, $x_{m+a}=y_{r+b}=1$ and $w_\x(m+a)=w_\y(r+b)$. 
Therefore, the packet $P_j$ will be placed on node $U_i$, where $i-(m+a)\equiv 0\ (mod\ n)$, $j-(r+b)\equiv 0\ (mod\ \theta)$, $i\in\{1,2,\ldots,n\}$ and $j\in\{1,2,\ldots,\theta\}$.	
For each integer $p=1,2,\ldots,n$ such that $i\neq p$, $A_p(j)+\sum_{s=1,s\neq j}^\theta A_i(s)A_p(s)=\sum_{s=1}^\theta A_i(s)A_p(s)\leq 1$, where $A_i(j)=1$ for $j\in\{1,2,\ldots,\theta\}$ such that $j-(r+b)\equiv 0\ (mod\ \theta)$.
From the Lemma \ref{U G Condition Flower code lamma}, the Lemma holds.

\subsubsection{Proof of Lemma \ref{rho 2 example}}
For a periodic binary sequence $\x$ with the period $\tau\leq n$, if $\ell=\left\lceil\frac{2\theta\tau}{w_\x(\tau)}\right\rceil$ then the entry sum will be $2\theta$.
For the particular periodic sequence, if $gcd(n,\theta)=1$ and $gcd(n,\tau)=1$ then any two node will not contain more than one common packet. 
Hence, the result holds. 

\subsubsection{Proof of Lemma \ref{ha ha ha}}
The proof follows the Theorem \ref{U G Condition Flower code equation} for the packet dropping sequence $\x=1^{w_\y}$.

\end{document}